%% file: subgraphs.tex
\documentclass[11pt]{article}

\usepackage{amsmath}
\usepackage{amsfonts}
\usepackage{amssymb}
\usepackage[margin=1in,a4paper]{geometry}
\usepackage{amsthm}
\usepackage{color}
\usepackage[pdftex]{graphicx}
\usepackage{microtype}
\usepackage{wrapfig} 

\usepackage{authblk}

\newcommand{\BO}[1]{{ O}\left(#1\right)}
\newcommand{\BTO}[1]{\tilde{ O}\left(#1\right)}

\newcommand{\BT}[1]{{\Theta}\left(#1\right)}
\newcommand{\BOM}[1]{\Omega\left(#1\right)}

\newtheorem{lemma}{Lemma}
\newtheorem{theorem}{Theorem}

\newcommand{\E}[1]{\mathbb{E}\left[#1\right]}

\newcommand{\Var}[1]{\text{Var}\left(#1\right)}

\renewcommand{\Pr}[1]{\text{Pr}\left(#1\right)}

\title{Subgraph enumeration in massive graphs\thanks{Part of this work was done while the author was working at the University of Padova. It is supported in part by University of Padova project CPDA121378,   MIUR of Italy project AMANDA,  and by the European Research Council  grant 614331.
}}

\author{Francesco Silvestri}
\affil{IT University of Copenhagen, Denmark\\ \texttt{fras@itu.dk}
}
\date{}

\begin{document}
\pagestyle{plain}
\maketitle

\begin{abstract}
We consider the problem of enumerating all instances of a given pattern graph in
a large data graph. Our focus is on determining the input/output (I/O)
complexity of this problem. Let $E$ be the number of edges in the data graph,
$k=\BO{1}$ be the number of vertices in the pattern graph, $B$ be the block
length, and $M$ be the main memory size. The main results of the paper are two
algorithms that enumerate all instances of the pattern graph. The first one is a
deterministic algorithm that exploits a suitable independent set of the pattern
graph of size $1\leq s \leq k/2$ and requires
$\BO{E^{k-s}/\left(BM^{k-s-1}\right)}$ I/Os. The second algorithm is a
randomized algorithm that enumerates all instances in
$\BO{E^{k/2}/\left(BM^{k/2-1}\right)}$ expected I/Os; the same bound also applies with high probability under some assumptions. A lower bound shows that 
the deterministic algorithm is optimal for some  pattern
graphs with $s=k/2$ (e.g., paths and cycles of even length, meshes of even
side), while the  randomized algorithm is optimal for a wide class of  pattern graphs, called Alon class (e.g., cliques, cycles and every graph with a perfect matching).
\end{abstract}

\section{Introduction}

This paper targets the problem of enumerating all subgraphs of an input
\emph{data graph} that are isomorphic to a given \emph{pattern graph}. Subgraph
enumeration is a tool for analyzing the structural and functional properties of networks (see, e.g.,~\cite{KairamWL12,GregoriLM13}), and  typical pattern graphs are cliques (e.g.,
triangles), cycles and paths. Subgraph enumeration is also strictly related to
the evaluation of conjunctive queries or multiway joins on a single large
relation~\cite{AfratiDU13}. 

The aim of this paper is to assess the input/output
(I/O) complexity of the enumeration problem when the data graph does not fit in the main memory.   
The main results of the paper are  external memory (EM)  algorithms for subgraph enumeration. 
In particular, we provide a deterministic algorithm which exploits
 a \emph{matched independent set} (MIS) of the pattern graph $H$, which is an independent set $S$ such that each vertex in $S$ can be  matched with a vertex not in $S$.
Let $E$ be the number of edges in the input data graph,
$k=\BO{1}$ be the number of vertices in the pattern graph, $B$ be the block
length, and $M$ be the main memory size. Our results are the following:

\begin{enumerate}
\item We give a deterministic algorithm for subgraph
enumeration that exploits a MIS $S$ of the pattern graph of size $s=|S|$, with $1\leq s \leq k/2$.
Its I/O complexity is $\BO{( E^{k-s} \log_M E )/(B M^{k-s-1})}$. As
an example, let $M=\BOM{E^\epsilon}$ for some constant $\epsilon>0$:  
we get $\BO{E^{k-1}/(B M^{k-2})}$ I/Os if the pattern graph is a $k$-clique ($s=1$),
and $\BO{E^{k/2}/(B M^{k/2-1})}$ I/Os if the pattern graph  is an even length path or
cycle, or a mesh of even side ($s=k/2$).

\item We propose a randomized algorithm for subgraph
enumeration. It exploits the random coloring technique in~\cite{PaghS13} for
decomposing the problem into smaller subproblems that are solved with the above deterministic algorithm. Its expected I/O complexity is
$\BO{E^{k/2}/\left(B M^{k/2-1}\right)}$. 
We show that the claimed I/O
complexity is also achieved with high probability when $M=\BOM{\sqrt{E}\log E}$ by adjusting
  the coloring process. 
We remark that the deterministic algorithm is a crucial component of the randomized one, and cannot be replaced by state-of-the-art techniques without increasing the I/O complexity.

\item We discuss some related issues.
We first show  that the enumeration of $T$ instances of
a pattern graph in the Alon class~\cite{AfratiSSU13} requires, even in the best case,
$\BOM{{T}/\left({B M^{ k/2-1}}\right)+{T^{2/k}}/{B}}$ I/Os. 
The Alon class includes important graphs like
cliques, cycles and, more in general, every graph with a perfect matching.
This lower bound  implies that the randomized algorithm is optimal in the worst case since a clique with $\sqrt{E}$ vertices  contains $T=\BT{E^{k/2}}$ instances of any pattern graph. It also shows that  the deterministic algorithm is optimal for some sparse pattern graphs (e.g.,  even length paths and cycles, 
meshes of even side) if $M=\BOM{E^\epsilon}$ for some constant $\epsilon>0$.
Finally, we  analyze  the work complexity of our algorithms: 
for pattern graphs in the Alon class, the deterministic and randomized algorithms require respectively
$\BTO{E^{k-s}/M^{k/2-s}}$ and $\BTO{E^{k/2}}$ total work, where the last term is just a polylog factor from the optimal bound.
\end{enumerate}

The assumption $k=\BO{1}$ is quite natural since it covers the most relevant case; however, the analyses of our algorithms do not assume $k$ to be constant and clearly state the dependency of the I/O complexities on $k$. 
Moreover, this paper focuses on  the enumeration of edge-induced
subgraphs which are isomorphic to the
pattern graph; however, we claim that our algorithms can be extended even to the enumeration of
vertex-induced subgraphs (see Appendix~\ref{app:induced} for more details). 

We do not require our algorithms to \emph{list} all instances of the pattern graph,
that is to store all instances on the external memory. We simply consider
algorithms that \emph{enumerate} instances: that is, for each instance, they
call a function \texttt{emit}$(\cdot)$ with the instance as input
parameter. 
Nevertheless, our upper and
lower bounds can be easily adapted to list all instances  by increasing the I/O
complexity of an unavoidable additive  $\BT{T/B}$ factor, where $T$ is the
number of instances.

\section{Related work and comparison with our results}
To the best of our knowledge, this is the first paper to deal
with the I/O complexity of the enumeration of a generic pattern graph.
Previous works have targeted  the I/O complexity of triangle enumeration. An optimal algorithm requiring
$\BO{\text{sort}(E)}$ I/Os for graphs with constant arboricity is given in~\cite{GoodrichP11}; this algorithm
however does not efficiently scale with larger arboricity. 
The works ~\cite{ChuC12,HuTC13} propose algorithms for a generic data graph incurring $\BO{E^2/(BM)}$ I/Os. 
In the special case where the
pattern graph is a triangle, our deterministic algorithm recalls the one
proposed in~\cite{HuTC13}, but it does not need to manage in a different way vertices of the data graph with degree $\leq {M}$ and with degree $>M$. 
The previous bound is improved to
an optimal $\BT{E^{3/2}/\left(B\sqrt{M}\right)}$ (expected) I/O complexity in~\cite{PaghS13,HuTQ15}, which respectively provide  randomized and deterministic algorithms.
Our randomized algorithm extends to a generic pattern graph the random vertex coloring technique introduced in~\cite{PaghS13}. 
However, this paper substantially differs from~\cite{PaghS13} since novel and non-trivial results are proposed: besides specific
technicalities required for the generalization of the coloring technique, we give the new deterministic
algorithm based on a MIS, which
is crucial for solving small subproblems generated by the
coloring technique, and we show that the I/O complexity of the randomized
algorithm holds even with high probability. 

An algorithm for the enumeration of $k$-cliques, for a given $k\geq 3$, is given in~\cite{ChibaN85} for the RAM model, but it requires
$\BOM{E^{k/2}/B}$ I/Os in a memory hierarchy. Multiway-join is a problem from
database theory   related to subgraph enumeration: however, the most relevant
algorithms (e.g.,~\cite{NgoPRR12}) ignore the memory hierarchy and do not
efficiently translate into our settings (a generous analysis would give
$\BOM{E^{k/2}/B}$ I/Os).
Algorithms for detecting the existence of a given pattern graph and/or for
counting the number of its instances have also been widely studied (e.g.,~\cite{KolountzakisMPT12,KaneMSS12,WilliamsW13}). 
However, these works rely on techniques (e.g., sampling,  sketches, fast matrix multiplication) that  allow to detect/count instances without explicitly materializing them, and hence cannot be used for enumeration. 

Subgraph enumeration has also been targeted in MapReduce. 
An algorithm for clique enumeration is given in~\cite{FinocchiFF14}, but it does not translate into an I/O efficient algorithm since subproblem size cannot
be tuned to fit  internal memory (unless $M=E$).
Triangle and general pattern graph enumerations are target in \cite{AfratiSSU13,SuriV11} and in~\cite{AfratiDU13}, respectively. Although these results are based on partitioning techniques similar to the one used
by our randomized algorithm, they assume a random input and
provide weak bounds with an arbitrary input. Better worst case bounds are provided in~\cite{ParkSKP14} by exploiting  the random partitioning in~\cite{PaghS13}.

We remark that the I/O complexity of previous algorithms for general subgraph enumeration is $\BOM{E^{k/2}/B}$, which becomes a
performance bottleneck (i.e., it dominates the $\BO{E^{k/2}}$ work complexity) as soon as reading a memory block in external memory is $\BOM{B}$ times slower than a CPU operation. In contrast, our randomized algorithm requires a smaller amount of I/Os without increasing the work complexity, and  avoids the I/O performance bottleneck even for slower external memories (i.e., until an I/O requires $\BO{BM^{k/2-1}}$ CPU operations). 

\section{Preliminaries}
\subsection{Models} 
We study our algorithms in the \emph{external memory model}, which has been
widely adopted in the literature (see, e.g., the survey by
Vitter~\cite{Vitter08}). The model consists of an internal memory of $M$ words
and an external memory of unbounded size. The processor can only use data stored
in the internal memory and move data between the two memories in blocks of
consecutive $B$ words. We suppose each vertex and edge to require one memory
word. The \emph{I/O complexity} of an algorithm is defined as the number of
input/output blocks moved between the two memories by the algorithm. Our
algorithms are aware of the memory hierarchy parameters, and  can be
straightforwardly adapted to a memory-cache hierarchy with an automatic replacement policy (e.g., LRU).
 
\subsection{Notation}
We denote with $G=(V, E)$ the simple and undirected input data graph. For
notational convenience, whenever the context is clear we use $E$ as a shorthand for the size of  set $E$ (and similarly for
other sets). We denote with $\deg(v)$ the degree of a vertex $v\in V$. We 
assume that the sizes of $V$ and $E$ are known, that all vertices in $V$ are
labeled with an unique identifier, and that the edge set $E$ is represented with
adjacency lists which are stored in consecutive memory positions 
and sorted by identifier. We observe that these assumptions can be guaranteed by
suitably sorting and scanning the input edges without asymptotically affecting
the I/O complexity of our algorithms.
\begin{wrapfigure}{r}{0.25\textwidth}
\centering \scalebox{.5}{\input{mis.tex}}
  \caption{The MIS for a $3\times 3$ mesh. Grey nodes denote the MIS; dashed lines are the probe edges. Since $k=4$, the probe of $h_5$ is $h_3$, which is not in the MIS.\label{fig:pattern}}
\end{wrapfigure}

We denote with $H=(V_H, E_H)$ the simple and undirected pattern graph
that we are looking for in the input graph $G$. Let $k=\vert V_H\vert$ and 
$V_H=\{h_1, \ldots, h_ k\}$. An \textit{instance} of $H$ in $G$ is a tuple
($v_1,\ldots, v_k$) of $ k$ distinct vertices of $G$ such that $(v_i, v_j)\in E$
for each edge $(h_i,h_j)\in E_H$. An instance is \textit{induced} if $(v_i,
v_j)\in E$ if and only if $(h_i,h_j) \in E_H$. Namely, {instances} are
edge-induced subgraphs of $G$, while {induced instances} are vertex-induced
subgraphs of $G$. For a given instance we say that vertex $h_i$ (resp., edge
$(h_i,h_j)$) is \textit{mapped} onto $v_i$ (resp., $(v_i,v_j)$). An instance is
enumerated by calling a function \texttt{emit}$(v_1,\ldots, v_k)$, and each call 
performs no I/Os and requires $\BO{1}$ operations.

We define a \emph{matched independent set $S$} (MIS)  of the pattern graph $H$  to be an independent set of $H$ for which exists in $E_H$ a matching between the $s$
vertices in $S$ and $s$ vertices in $V_H \setminus S$, with $s=|S|$. We
have $1\leq s \leq k/2$. The maximum size of a MIS is $s=\lfloor k/2\rfloor$ for a
cycle of length $k$ or a mesh of size $\sqrt{k}\times \sqrt{k}$, while it
is $s=1$ for a $k$-clique. 
For a given MIS $S$, we let $h_{
k-s+1},\ldots, h_{ k}$ denote the vertices of $H$ in $S$ and assume that
$h_i\in V_H\setminus S$ is matched with $h_{k-s+i}\in S$ for every $1\leq i
\leq s$.
Finally,  we define the
\emph{probe vertex} of vertex $h_i$, with $1\leq i \leq k-s$ as follows:
it is $h_{k-s+i}$ if $1\leq i \leq
s$ (i.e., a vertex in $S$ is the probe vertex of its companion in the matching);
otherwise it is an arbitrary neighbor vertex in $S$ if $s+1 \leq i\leq k-s$. 
If $h_j$ is the probe vertex of $h_i$, then we say that the \emph{probe index
of $i$}, denoted with $P(i)$, is $j$ and that the \emph{probe edge} of $h_i$ is
$(h_i, h_j)$ (see example in Figure~\ref{fig:pattern}).
Since we are interested in  pattern graphs with a very small number of nodes, we suppose that an exhaustive search on the pattern graph is used to find a MIS with the largest size;  we leave as open problem to derive an efficient algorithm for extracting a large MIS.

\section{Deterministic EM Algorithm}\label{sec:det}

In this section we describe the deterministic algorithm for enumerating all
instances of the pattern graph $H$ by exploiting a MIS $S$
of $H$. The algorithm works for any $S$, however the best performance are
reached when $S$ is the maximum MIS.
For the sake of simplicity, we  assume that $s<k/2$, and hence
that there exists at least one vertex in $V_H\setminus S$, say $h_{k-s}$, not 
matched with a vertex in $S$. The case $s=k/2$, covered in Appendix~\ref{app:detalgspecialcase}, is based on the same approach but requires some minor
technicalities that increase the I/O complexity by a multiplicative factor
$\BO{\log_M E}$. This factor is asymptotically negligible as soon as 
$M=\BOM{ E^\epsilon}$ for some constant $\epsilon>0$.
We first provide a simple high level explanation of the algorithm, and then give a more detailed description. 

We observe that an instance of $H$ in $G$ is uniquely defined by the mapping of the $k-s-1$ probe edges  associated with $h_1, \ldots, h_{k-s-1}$ and of  vertex $h_{k-s}$, since such a 
mapping automatically fixes the mapping of all vertexes of $H$. 
As an example consider again Figure~\ref{fig:pattern}: any instance of the pattern graph is univocally given by the mapping of the probe edges $(h_1,h_6),(h_2,h_7),(h_3,h_8),(h_4,h_9)$ and of vertex $h_5$.
The opposite direction is not true: a mapping may not denote an instance of $H$ since a non-probe edge of $H$ may be mapped on an edge not in $G$.
The deterministic algorithm exploits these facts: it generate all mappings of $k-s-1$ probe edges and of  vertex $h_{k-s}$, and then verifies which mappings  denote real instances of $H$ in $G$.
The generation of all mappings is done with an I/O-efficient  exhaustive search.

We assume that the edges of $G$ are split into $\phi=\BT{Ek/ M}$ chunks. Specifically, the adjacency lists of $G$ are split into $\phi$ consecutive chunks $C_i$ of size in the range $(M/(8k),M/(4 k)]$, where 
$1\leq i \leq  \phi$ and $\phi\in [4 k E/M, 8kE/M)$. A vertex whose
adjacency list is completely contained in a chunk is called \textit{complete},
and  \textit{incomplete} otherwise. We require each chunk to contain at most
one incomplete vertex. 
It can be proved that such a partition exists and can be constructed by scanning the edge set $E$.

The algorithm works in $\phi^{k-s-1}$ \emph{rounds}, which run over all possible ways of selecting (with repetitions) $k-s-1$ chunks from $\phi$ chunks. In each round, the following operations are done (step numbers refer to the pseudocode in the next page):  the $k-s-1$ selected chunks are loaded into internal memory (steps 1-2);  by scanning  the entire edge list of $G$, all edges connecting two incomplete vertexes of the loaded chunks are inserted in memory, if not already in a chunk (step 3); finally, all instances of $H$ where the $i$-th probe edge is mapped on an edge in the $i$-th chunk are enumerated (step 4).
This last operation proceeds in \emph{iterations} that run over all possible ways of mapping $h_{k-s}$ to a vertex
$v\in V$ (note that this vertex is not fixed by the mapping of probe edges).
In each iteration (steps 4.a-4.c), the algorithm scans the adjacency list of  $v$ and checks if there exists an instance where $h_{k-s}$ is mapped on $v$ and the $i$-th probe edge is mapped on an edge in the $i$-th chunk;  function \texttt{emit}$(\cdot)$ is called for each existing instance.

We now provide a more detailed description of the deterministic algorithm. 
Consider a generic round and denote with  $C_{\ell_1},\ldots, C_{\ell_{ k-s-1}}$, for suitable
values of $\ell_1,\ldots, \ell_{ k-s-1}$,  the $k-s-1$ selected chunks. The algorithm uses the support sets $E'$, $E''$, $E_i$ for each $1\leq i
\leq k-s-1$,  and $V_i$ for each $1\leq i \leq k$, which we suppose to be stored in
internal memory and initially empty. Each round performs the following
operations:
\begin{enumerate}
\item \label{s1} For each $1\leq i \leq k-s-1$, we load in memory 
$C_{\ell_i}$, and fill $V_i$ and $E_i$ with the vertexes and edges that are contained in  $C_{\ell_i}$. Specifically, we  add to  $V_i$ all vertexes whose adjacency list is (partially) contained in $C_{\ell_i}$, and add to $E_i$ all edges $(u,v)$ where $u\in V_i$ and $(u,v)$ appears in the (part of) adjacency list of $u$  in $C_{\ell_i}$.

\item \label{s2} For each $1\leq i\leq s$, we add to $V_{k-s+i}$ all vertexes of $G$ on which $h_{k-s+i}\in S$ can be mapped assuming
that the probe edge of $h_i$ is mapped onto an edge in $E_i$. Formally, each vertex $u\in V$ is added to $V_{k-s+i}$ if and only if there exists a vertex $v\in V_i$ such that $(v,u)\in E_i$. No I/Os are needed in this step since the
operation can be performed by reading the chunks in internal
memory.\footnote{Note that at this point all sets $V_i$, with $i \neq k-s$, are
not empty because $s<k/2$. Indeed, when $s=k/2$,  $V_k$ is not
filled since $h_k$ is the probe vertex of $h_{k/2}$.}

\item \label{s3} Edge set $E'$ is filled with all edges of $G$ connecting
vertices in $(\cup_{i=1}^{k-s-1} V_i) \cup (\cup_{i=k-s+1}^k V_i)$ that are
not already available in internal memory but are required for correctly
enumerating instances. Formally, for each $(h_i,h_j)\in E_H$ with $1\leq
i,j\leq k$
and $i,j\neq k-s$, each edge $(v,v')\in E$ is added to $E'$ if and only if $v\in
V_{i}$,  $v'\in V_{j}$, but $(v,v')\notin E_{i}\cup E_{j}$. (We note that an
edge can be added to $E'$ although it is contained in $E_l$ for some $l\neq
i,j$.) This operation can be performed by scanning once the adjacency lists of
$G$. 

\item \label{s4} Enumerate all instances of $H$ in $G$ where vertex
$h_i$ is mapped onto a vertex in $V_i$ and its probe edge onto edges in $E_i$,
for any $0\leq i <k-s$. The enumeration proceeds in $V$ iterations. In an
iteration, we set $V_{k-s}=\{v\}$, for any possible value of $v\in V$, and then
the following operations are done:

\begin{enumerate}

\item \label{s4a} Let $E''$ be the edge set containing all edges between $v$ and
vertices in $(\cup_{i=1}^{k-s-1} V_i) \cup (\cup_{i=k-s+1}^k V_i)$ which are not
already in internal memory. Formally,  each edge $(v,v')$ is added to $E''$ if
and only if $v'\in V_i$ but $(v,v')\notin E_i$. (We note that an edge can be
added to $E''$ although it is contained in $E_l$ for some $l\neq i$.) This step
requires a scan of the adjacency list of $v$.

\item \label{s4b} Using a naive approach (see Section~\ref{sec:ext}), enumerate
in main memory all instances of $H$ in the subgraph $(\cup_{i=1}^{k} V_i, E'\cup
E'' \cup (\cup_{i=1}^{k-s-1} E_i))$ of $G$ where vertex $v_{k-s}$ is mapped onto
$v$, and the probe edge of $h_i$ is mapped onto an
edge in $E_i$ for each $1\leq i \leq k-s-1$.

\item \label{s4c} Empty sets $V_{k-s}$ and $E''$. 

\end{enumerate}

\item \label{s5} Empty sets $E'$, $V_{i}$ for each $1\leq i \leq k$, and $E_{i}$
for each $1\leq i \leq k-s-1$.

\end{enumerate}

Correctness and I/O complexity are stated in the following theorem:

\begin{theorem}\label{th:det}
The above algorithm correctly enumerates all instances of a given pattern graph
$H$ and its I/O complexity is $\BO{(8k)^{k-s-1}\frac{E^{k-s}}{B M^{k-s-1}}}.$
\end{theorem}
\begin{proof} \emph{(Sketch)}
In order to prove the correctness of the algorithm, it is necessary to prove that all instances are emitted once. 
As already mentioned all instances are uniquely defined by the mapping of the probe edges  of $h_1,\ldots, h_{k-s-1}$ and of the  vertex $h_{k-s}$. 
Standard combinatorial arguments show that each one of these mappings is generated once during the execution of the algorithm. 
The scan of $E$ performed at the beginning of each round and the scanning of the adjacency list of vertex $v$ at the beginning of an iteration, guarantee that all edges necessary for verifying that a mapping gives a correct instance of $H$ in $G$  are available in the internal memory. 
The amount of internal memory  used in each round is at most $M$ since there are $k-s-1$ chunks of size at most $M/(4k)$ and at most $\BO{k^2}$ edges are added in steps 3 and 4.a. 
The naive enumeration in step 4.b then does not require any I/O. 
The I/O complexity of each round is therefore dominated by the two scans of the adjacency lists of $E$ (in step 3, and in the $V$ iterations of step 4.a). Since there are $\phi^{k-s-1} \leq (8k E/M)^{k-s-1}$ rounds,  the claim follows.
\qed
\end{proof}
\begin{proof}
We first prove the correctness of the algorithm. 
Consider an instance $(v_1,\ldots, v_ k)$ of the pattern graph $H$ in $G$. For
each $1\leq i\leq k-s-1$, let $C_{\ell_i}$ be the chunk containing $(v_i,
v_{P(i)})$ with $v_i\in V_i$ (we recall that $P(i)$ is the probe index of $i$,
that is $h_{P(i)}$ is the probe vertex of $h_i$). Consider the unique round
where chunks $C_{\ell_1},\ldots, C_{\ell_{ k-s-1}}$ are loaded in memory in this
order. Then, $(v_1,\ldots, v_ k)$ is correctly enumerated in the iteration where
$V_{ k-s}$ is set to $v_{ k-s}$. Indeed, all vertices and edges are available in
internal memory: Step~\ref{s1} guarantees that $v_i\in V_i$ for $1\leq i \leq
k-s-1$; Step~\ref{s2} adds $v_{k-s+i}$ to $V_{k-s+i}$ for $1\leq i \leq s$ since
the edge $(v_i, v_{k-s+i}) \in E_i$ by assumption and $P(i)={k-s+i}$ (we note
that this would not happen for $V_{k}$ when $s=k/2$ since $V_{k/2}$ is empty at
this point); Step~\ref{s3} we have that all edges connecting vertices in $\{v_1,
\ldots, v_{k-s-1}, v_{k-s+1},\ldots, v_k\}$ are in memory (more specifically,
all edges between complete vertices are already in internal memory after
Step~\ref{s1}); finally, Step~\ref{s4a} guarantees that all edges between
$v_{k-s}$ and $\{v_1, \ldots, v_{k-s-1}, v_{k-s+1}\ldots v_k\}$ are in memory.
The instance  $(v_1,\ldots, v_ k)$  is enumerated once: indeed, the instance can
be enumerated only in the unique round where chunks $C_{\ell_1},\ldots,
C_{\ell_{ k-s-1}}$ are loaded in memory in this order (a different order may
enumerate an automorphism but not the same instance), and in the unique
iteration where $V_{k-s}$ is set to $v_{k-s}$ (clearly, the naive approach for
enumeration in Step~\ref{s4b} must emit  each instance once).

We now show that the total amount of required internal memory is at most $M$.
The sets $V_i$ and $E_i$, for each $i\neq k-s$, have sizes at most $M/4k$ each,
and thus at most $M(k-1)/(2k)$ memory words are required (note that chunks
$C_{\ell_1},\ldots, C_{\ell_{ k-s-1}}$ can be removed from the internal memory
after Step~\ref{s1}). The size of $V_{k-s}$ is clearly one memory word. The size
of $E'$ is at most $(k-1)^2$ words: indeed, an edge $(v,v')\in E$ is added to
$E'$ if and only if $v\in V_{i}$,  $v'\in V_{j}$, and $(v,v')\notin E_{i}\cup
E_{j}$; this implies that $v$ and $v'$ are incomplete vertices, otherwise
$(v,v')$ would be in $E_{i}\cup E_{j}$; then, being at most one incomplete
vertex per chunk, the claim follows. Similarly, we have that $E''$ has size
at most $(k-1)$ words. Then, the total amount of space is $M(k-1)/(2k)+k^2$
which is not larger than $M$ since $k<<M$.

Finally, we  analyze the I/O complexity of the algorithm. The I/O cost for
enumerating instances in Step~\ref{s4b} is negligible since the problem fits in
memory and all operations are  performed in main memory. Then the I/O complexity
of each round is  asymptotically upper bounded by a constant number of scans of
the whole edge set $E$. Since there are $\phi^{k-s-1}\leq (8  k E/ M)^{k-s-1}$
rounds, the claimed I/O complexity follows.
\qed
\end{proof}

\section{Randomized EM Algorithm}\label{sec:rand}

We  are now ready to introduce the randomized algorithm. 
The algorithm, by making use of the random coloring technique in~\cite{PaghS13},
decomposes the problem  into small subproblems of expected size $\BO{M}$, which
are then solved with the previous deterministic algorithm.  We
assume that the maximum degree of $G$ is $\sqrt{EM}$; however, in Section~\ref{sec:degree}, we show how this assumption can be removed by increasing the I/O complexity by a multiplicative factor $k^{\BO{k}}$.
We first prove the expected I/O complexity and then show how to get the high probability under some assumptions in Section~\ref{sec:whp}.

Let $\xi: V \rightarrow \{1, \ldots, c\}$, with $c=\sqrt{E/M}$, be a vertex
coloring chosen uniformly at random from a family of $2(k-s+1)$-wise independent
family of functions.
The coloring $\xi$ partitions the edge set $E$ into $c^2$ sets of expected size
$M$. 
For each pair of colors $\tau_1,\tau_2\in \{1, \ldots, c\}$ and $\tau_1\leq
\tau_2$, we denote with $E_{\tau_1,\tau_2}$ the set containing edges colored
with $\tau_1$ and $\tau_2$, that is $E_{\tau_1,\tau_2}=\{(u,v)\in E \vert
\min\{\xi(u),\xi(v)\}=\tau_1, \max\{\xi(u),\xi(v)\}=\tau_2\}$.
Each instance $(v_1,\ldots, v_k)$ of the pattern graph can be colored by $\xi$
in
$c^k$ ways, and it is said to be \emph{$(\tau_1,\ldots,\tau_{k})$-colored} if
$\xi(v_i)=\tau_i$ for each $1\leq i\leq k$. 

The randomized algorithm enumerates all instances by decomposing the problem
into $c^k$ subproblems. Each subproblem finds all
$(\tau_1,\ldots,\tau_{k})$-colored instances according to a given $k$-tuple of
colors using the previous deterministic algorithm on the edge set
$\cup_{\tau_i\leq \tau_j} E_{\tau_i,\tau_j}$.
The algorithm is organized as follows:
\begin{enumerate}
\item Randomly select a coloring $\xi$ from a $2(k-s+1)$-wise independent 
family
of functions.
\item Using sorting, store edges in  $E_{\tau_1,\tau_2}$ in consecutive
positions, for each color pair $(\tau_1,\tau_2)$.
\item For each $k$-tuple of colors $(\tau_1,\ldots,\tau_{k})$, enumerate all
$(\tau_1,\ldots,\tau_{k})$-colored instances using the algorithm in
Section~\ref{sec:det} on the sets $E_{\tau_i,\tau_j}$, for each $\tau_i\leq
\tau_j$.
\end{enumerate}

In order to bound the I/O complexity of the randomized algorithm, we introduce the following technical lemma that upper bounds the expected number
$X_t$ of possible tuples of $t$ edges in $E$ that are colored in the same way by $\xi$. 
A closed form of this quantity is $X_t=\sum_{\tau_1\leq
\tau_2, E_{\tau_1,\tau_2}\geq t} \frac{E_{\tau_1,\tau_2}!}{(E_{\tau_1,
\tau_2}-t)!}$ 
(note that sets $E_{\tau_1,\tau_2}$ with less than $t$ edges do not contribute).

\begin{lemma}\label{lem:fact}
Let $\xi:V\rightarrow \{1,\ldots, c\}$ be chosen uniformly at random  from a
$2t$-wise independent family of hash functions, where $c=\sqrt{E/M}$.  If
$M=\BOM{t^2}$ and the maximum vertex degree in $G$ is $\sqrt{EM}$, then
$\E{X_t}\leq (2t)^{t-1}EM^{t-1}$.
\end{lemma}
\begin{proof}
\newcommand{\e}{\mathbf{e}}
We prove the claim by induction on $t$. The claim is verified for $t=1$ since
$\E{X_1}=E$.
For each tuple $\e=(e_1,\ldots, e_t)$ of $t$ distinct edges in $E$ and for each
$2\leq i\leq t$, let $Y^\e_i=1$ if $e_i$ is in the same set $E_{\tau_1,
\tau_2}$, for some colors $\tau_1, \tau_2$, of edges $e_{1},\ldots, e_{i-1}$,
and 0 otherwise. Set $Y^\e_1=1$. We get  $X_t=\sum_{\e} Y^{\e}_t$. 
Since there are at most $2t$ vertices and $\xi$ is $2t$-wise, we get
$$
\Pr{Y^\e_t=1}\leq \left\{
\begin{array}{ll}
\Pr{Y^\e_{t-1}=1} /c^2 & \text{if $e_t$ is not adjacent to $e_1,\ldots,
e_{t-1}$}\\
\Pr{Y^\e_{t-1}=1}/c & \text{if $e_t$ is  adjacent to $e_1,\ldots, e_{t-1}$ on
one vertex}\\
\Pr{Y^\e_{t-1}=1} & \text{if $e_t$ is adjacent to $e_1,\ldots, e_{t-1}$ on two
vertices}
\end{array}
\right.
$$
Each $(t-1)$-tuple $\e'$ can be extended by at most $E$ edges that are not
connected with $\e'$, or by $2(t-1)\sqrt{EM}$ edges that are connected to $\e'$
on just one vertex (recall that the maximum degree of a vertex is $\sqrt{EM})$,
or by $(t-1)(2t-3)$ edges  that are connected to $\e'$ on two vertices.
Therefore, we get 
\begin{align*}
\E{X_t}&=\sum_{\e} \Pr{Y^\e_t=1}
\leq \E{X_{t-1}}\left(\frac{E}{c^2} + 2(t-1)\frac{\sqrt{EM}}{c}+
(t-1)(2t-3)\right).
\end{align*}
Since the right term is upper bounded by $2tM \E{X_{t-1}}$, the lemma
follows.\qed
\end{proof}

We are now ready to show the correctness and  I/O complexity  of the randomized algorithm.

\begin{theorem}\label{thm:randalg}
The above randomized algorithm  enumerates all instances of a given pattern
graph
$H$. 
If the maximum vertex degree of $G$ is $\sqrt{EM}$, then the expected I/O
complexity of the algorithm is 
$
\BO{(8k)^{4(k-s+1)}{  E^{ k/2}}/({B M^{k/2-1}})}.
$
\end{theorem}
\begin{proof}
 The correctness easily follows since each instance is colored with a suitable
color tuple $(\tau_1,\ldots,\tau_k)$ and is enumerated only in the subproblem
associated with this color tuple.
The cost of each subproblem is given by Theorem~\ref{th:det}, however for
simplicity, we upper bound the cost of the deterministic
algorithm with $\BO{(8k)^{k-s} E^{k-s+1}/(BM^{k-s})}$ in order
to get rid of the logarithmic term. The I/O complexity
$Q(E)$ of the algorithm is upper bounded by the sum of the costs of all $c^k$
subproblems. Then,
\begin{align*}
Q(E) =& \BO{\frac{ (8k)^{ k-s}}{B M^{k-s}}\sum_{(\tau_1,\ldots,
\tau_{k})}{\left(\sum_{\tau_i\leq \tau_j} E_{\tau_i,\tau_j}\right)^{k-s+1}}}\\
\leq& \BO{\frac{(8k)^{2( k-s+1)}}{B M^{ k-s}}\sum_{(\tau_1,\ldots,
\tau_{k})}\sum_{\tau_i\leq \tau_j}E_{\tau_i,\tau_j}^{ k-s+1}}
\\
\leq & 
\BO{\frac{ c^{k-2} (8k)^{2(k-s+1)}}{B M^{k-s}}\sum_{\tau_1\leq
\tau_2}E_{\tau_i,\tau_j}^{ k-s+1}}\\
\leq &
\BO{\frac{ c^{k-2} (8k)^{3(k-s+1)}}{B M^{k-s}}\sum_{\tau_1\leq \tau_2,
E_{\tau_1,\tau_2}\geq k-s+1}\frac{E_{\tau_1,\tau_2}!}{(E_{\tau_1,\tau_2}-
k+s-1)!}}\\
\leq & 
\BO{\frac{ c^{k-2} (8k)^{3(k-s+1)}}{B M^{k-s}}X_{k-s+1}}.
\end{align*}
By the linearity of expectation, we get
$
\E{Q(E)}=\BO{\frac{ c^{k-2} (8k)^{3(k-s+1)}}{B M^{k-s}}\E{X_{k-s+1}}}.
$
Then, by Lemma~\ref{lem:fact} and  the $2(k-s+1)$-wiseness of $\xi$, we get the
claimed result.
\qed

\end{proof}

We remark that our deterministic algorithm   is crucial for getting the claimed
I/O complexity. Indeed, the algorithm used in the subproblems should require
$\BO{M/B}$ I/Os for solving subproblems of size $\BT{M}$ (note that subproblems
may not perfectly fit the memory size). 
Using  existing enumeration algorithms, which require $\BOM{M^{k/2}/B}$ I/Os for
solving subproblems of size $\BT{M}$, would increase the total I/O complexity by
a multiplicative factor $\BOM{M^{k/2-1}}$.

\subsection{Getting the high probability}\label{sec:whp} 
If $M=\BOM{\sqrt{E}\log E}$, the randomized coloring process  can be slightly modified to get with probability  $1-1/\BT{E}$ the  claimed I/O complexity. For the sake of simplicity we assume the maximum degree to be $\sqrt{EM}$, although it is possible to remove this assumption even for higher degree by adapting the procedure described in the next Section~\ref{sec:degree}.\footnote{For $k=\BO{1}$, the procedure in Section~\ref{sec:degree} consists in  repeating the randomized algorithm a constant amount of times. Then by an union bound, we get that the claimed complexity.}

A vertex $v\in V$ has \emph{high degree} if $ \sqrt{E}\leq \deg(v) \leq
\sqrt{EM}$ and has \emph{low degree} if $ \deg(v) < \sqrt{E}$. The coloring
process is modified as follows. The colors of low degree vertices are assigned
independently and uniformly at random. 
The colors of  high degree vertices are set by partitioning  vertices into $c$
groups so that the sum of degrees within each group is in  $[\sqrt{EM},
2\sqrt{EM})$, and then high degree vertices within the $i$-th group get color
$i$  (this operation requires  $\BO{1}$ sorts).  

Our argument relies on the technique by Janson~\cite[Theorem 2.3]{Janson04} for
obtaining a strong deviation bound for sums of dependent random variables, which
we recall here for completeness. Let  $X = \sum_{i=1}^p Y_i$ where each $Y_i$ is
a random variable with  $Y_i - \E{Y_i}\leq 1$, and let $\psi=\sum_{i=i}^p
\Var{Y_i}$. 
Denote with $\Delta$  the maximum degree of the dependency graph of $Y_1,
\ldots, Y_p$: this is a graph with vertex set $Y=\{1, \ldots, p\}$ such that if
$B \subset Y$ and  $i\in Y$ is not connected to a vertex in $B$, then $Y_i$ is
independent of $\{Y_j\}_{j \in B}$. Then, for any $d >0$, we have
$
\Pr{X \geq (1+d) \E{X}}\leq e^{-\frac{8 d^2 \E{X}^2}{25 \Delta (\psi+d \E{X}
/3)}}$. 

\begin{theorem}\label{th:whp}
Let $M=\BOM{\sqrt{E} \log E}$ and let the maximum vertex degree of $G$ be
$\sqrt{EM}$. Then, the I/O complexity of the
above algorithm is 
$\BO{(8k)^{6(k-s)}{E^{k/2}}/({B M^{k/2-1}})}$ with  probability at least $1-1/E$.
\end{theorem}
\begin{proof} Let $E^L$ be the set of edges in $E$ connecting two low degree vertices. We
also define $E^H=E/E^L$,  $E_{\tau_1, \tau_2}^L = E_{\tau_1, \tau_2} \cap E^L$,
$E_{\tau_1, \tau_2}^H = E_{\tau_1, \tau_2} \cap E^H$.
We first show that the size of $E_{\tau_1, \tau_2}^L$ for any color pair 
$\tau_1, \tau_2$ is smaller than $2M$ with probability at least $1-1/(2E)$.
Assume for simplicity that $\vert E^L\vert =\vert E\vert $.
For each edge $e\in E^L$, define the random variable   $Y_e$ to be $1$ if edge
$e$ is in $E_{\tau_1, \tau_2}^L$, and $0$ otherwise. 
We thus have $E_{\tau_1, \tau_2}^L = \sum_{e\in E} Y_e$.
Each random variable $Y_e$ depends on the at most $2\sqrt{E}$ variables
associated with edges adjacent to $e$, while it is independent of the remaining
ones.\footnote{Note that this is not the case if low degree vertices were
colored with $2(k-s)$-wise independent hash functions.}  
Since $Y_e-\E{Y_e}<1$, we use the aforementioned result by Janson by setting
$p=E$, $\E{E_{\tau_1, \tau_2}^L}=M$, $\psi=E(1/c^2-1/c^4)<M$, $d=1$, $\Delta=
2\sqrt{E}$. Then we get
$
\Pr{E_{\tau_1, \tau_2}^L \geq 2M}\leq e^{-\frac{4 M}{25 \sqrt{E}}}. 
$
By an union bound, the probability that  $E_{\tau_1,\tau_2}^L$ is smaller than $2M$ for every
color pair is at least
$1- c^2 e^{-\frac{4 M}{25 \sqrt{E}}}\geq 1-1/(2E)$ when $M=\BOM{\sqrt{E} \log
E}$.

We now show that the set $E_{\tau_1, \tau_2}^H$ has size $8M$  with probability
at least $1-1/(2E)$.
There are at most $2\sqrt{M}$ high degree vertices colored with a given color.
Then, there cannot be more than $4M$ edges connecting two high degree vertices
in $E_{\tau_1, \tau_2}^H$. Consider now the set $E^{H*}$ of edges connecting
high degree vertices of colors $\tau_1$ or $\tau_2$ to low degree vertices. We
have $E^{H*}\leq 4\sqrt{EM}$. 
For each $e\in E^{H*}$, define the random variable  $Y_e$ to be 1 if the low
degree vertex gets color $\tau_1$ or $\tau_2$, and $0$ otherwise. 
We have  $E_{\tau_1, \tau_2}^H\leq \sum_{e\in E^{H*}} Y_e$.
Since random variables may be dependent, we apply again the result by Janson
with $p=4\sqrt{EM}$, $\E{E^{H*}}=8M$, $\psi=\sum_{e\in E^{H*}}\Var{Y_e}\leq 8M$,
$d=1/2$, $\Delta =  2 \sqrt{E}$ (since only low degree vertices are randomly
colored). Then,
$
\Pr{E^H_{\tau_1,\tau_2} \geq  12 M}\leq e^{-\frac{2 M}{25 \sqrt{E}}}.
$
Then, the probability that  $E_{\tau_1,\tau_2}^H$ is smaller than $16M$ for
every color pair is at least
$1- c^2 e^{-\frac{2 M}{25 \sqrt{E}}}\geq 1-1/(2E)$ when $M=\BOM{\sqrt{E} \log
E}$.

Therefore, we have that each  $E_{\tau_1,\tau_2}$ has size at most $16M$ with
probability at least $1-1/E$. Since each subproblem  receives at most $ k^2$
edge sets, the I/O complexity of a subproblem is $\BO{(18k^2)^{k-s}(8k)^{4(k-s-1)}
M/B}$. Since there are $c^k$ subproblems, the claimed I/O complexity follows.
\qed
\end{proof}

It deserves to be noticed that it is possible to color low degree vertices with
a coloring from a $2(k-s)$-wise independent family and still get the claimed I/O
complexity  with probability $1-1/E^\epsilon$, for $0\leq \epsilon \leq 1/4$, as
soon as $M\geq E^{3/4+\epsilon}$. It suffices to use a technique by Gradwohl and
Yehudayoff~\cite[Corollary 3.2]{GradwohlY08} in our argument instead of the
aforementioned result by Janson~\cite[Theorem 2.3]{Janson04}. 

\subsection{Removing the degree assumption}\label{sec:degree}
Although the assumption in the randomized algorithm that the maximum degree in
$G$ is at most $\sqrt{EM}$ is reasonable  for real datasets, it can be removed
by increasing the I/O complexity by a multiplicative $k^{\BO{k}}$ factor.
We use the previous randomized algorithm as a black box and exploit a coloring
technique that should not be confused with the one used inside the randomized
algorithm. We denote with $V_H$ the set of \emph{very high degree} vertices in
$G$ (i.e., degree larger than $\sqrt{EM}$), and with $V_L = V\setminus V_H$ the
remaining low degree vertices. We let $G_L=(V_L, E_L)$ denote the subgraph of
$G$ induced by $V_L$.

Let $p\in[0,k]$. Consider the following simpler problem: enumerate all
instances of $H$ where $p$ given vertices of $H$, say for notational simplicity
$h_{1}, \ldots h_{p}$, are respectively mapped onto $p$ given very high degree
vertices $v'_1, \ldots, v'_p$, and where the remaining vertices  of $H$ are
mapped onto vertices in $V_L$. Since the mapping on the first $p$ vertices is
given we assume that if $(h_i,h_j)\in V_H$ then $(v'_i, v'_j)\in E$ for any
$1\leq i,j\leq p$ (this can be checked in scanning complexity).
We now show that this problem reduces to the enumeration in $G_L$ of a
suitable \emph{colored} pattern graph with $k'=k-p$ vertices, and which can be
solved with the previous randomized algorithm.
Suppose that each vertex in $V_L$ is colored with a $p$-bit color,
initially set to $0$. Then, for each $i\in [1,p]$ and for each vertex $v\in
V_L$ adjacent to $v'_i$, we
update the color of $v$ by setting the $i$-th bit to 1 (note that at the end of
this operation, a vertex color can have several bits set to 1). Define
the color tuple $d=(d_1, \ldots, d_{k'})$ as follows: set each term to 0; then,
for each $1\leq i \leq p$ and for each $h_{p+j}$ adjacent to $h_i$ in $H$, we
set the $i$-th bit of $d_j$ to 1.
Let $H'$ be the subgraph of $H$ induced by $h_{p+1},\ldots, h_k$.
Then, the problem can be solved by emitting instances $(v'_1, \ldots, v'_p,
v''_1,\ldots, v''_{k'})$, where $(v''_1,\ldots, v''_{k'})$ is every
instance of $H'$ in $G_L$ where vertices are colored according with $d$
(i.e., the $i$-th vertex of the instance has color $d_i$).
The colored instances of $H'$ can be obtained by
adapting the previous randomized algorithm to throw away
instances that are not compatible with coloring $d$.

By iterating the previous technique for any value of $p$ and for any matching of $p$ vertices in $H$ with $p$ very high degree vertices, we get the
claimed result.
\begin{theorem}\label{thm:rem}
The above algorithm  enumerates all instances of a given pattern
graph $H$ and  the expected I/O complexity  is 
$
\BO{k^{5(k-s+1)}{  E^{ k/2}}/({B M^{k/2-1}})}.
$
\end{theorem}
\begin{proof}
 We first show that the technique correctly enumerates all instances 
where $h_{1}, \ldots, h_{p}$ are
respectively mapped onto the very high degree vertices $v'_1, \ldots, v'_p$, and
where the remaining vertices  of $H$ are mapped onto vertices in $V_L$.
Consider the emitted  $(v'_1, \ldots, v'_p, v''_1,\ldots, v''_{k'})$ tuple.
We now prove that this instance satisfies the required properties. Clearly,
$v''_i\in V_L$ by construction. We now show that if $(h_i, h_j)\in H$ then the
edge is correctly mapped onto $E$. If $1\leq i,j\leq p$ or $p+1\leq i,j\leq k$,
the claim is verified, respectively,  by the initial assumption on $v'_1,
\ldots, v'_p$ and and by the correctness of the randomized algorithm. 
Suppose $1\leq i\leq p$ and $p+1\leq j\leq k$ (the opposite is equivalent).
Color $d_{j-p}$ must have the $i$-th bit set to $1$ since $(h_i, h_j)\in E_H$.
Since the instance must verify the coloring tuple, vertex $v''_{j-p}$ has
color $d_{j-p}$ and then it is adjacent to $v'_i$ since the $i$-th bit is
1.
Vice versa, it can be similarly shown that all instances that satisfy the
desired properties are correctly enumerated.

Let $r\leq 2\sqrt{E/M}$ be the number of very high degree vertices. Since for a
given $p$ the technique is called $r^p \frac{k!}{(k-p)!}$, the expected
I/O complexity can be upper bounded as follows:
$$
\BO{\sum_{p=0}^k r^p \frac{k!}{(k-p)!} k^{4(k-p-s+1)}
\frac{E^{(k-p)/2}}{BM^{(k-p)/2-1}}}=
\BO{k^{5(k-s+1)}\frac{E^{k/2}}{BM^{k/2-1}} }.
$$
\qed
\end{proof}

We note that the subsequent lower bound does not hold for
the technique proposed for getting rid of the degree assumption. Indeed,
information on graph connectivity are encoded in the coloring bits, but 
the lower bound requires at least one memory word for each vertex or edge.
However,  if $k$ is a small constant, the lower bound
still applies by using a memory word instead on a single bit.

\section{Further Extensions}\label{sec:ext}

\subsubsection{Lower Bound on I/O Complexity}

We now describe a lower bound on the I/O complexity for any algorithm
that enumerates $T$ instances of a pattern graph in the class of graphs named \emph{Alon class}~\cite{AfratiSSU13}.
A graph in the Alon class has the property that vertices can be partitioned into
disjoint sets such that the subgraph induced by each partition is either a
single edge, or contains an odd-length Hamiltonian cycle.
As in previous works~\cite{HuTC13,PaghS13} on triangle  enumeration, we assume
that each edge or vertex requires at least one memory word. That
is, at any point in time there can be at most $M$ edges/vertices in memory, and
an I/O can move at
most $B$ edges/vertices to or from memory. This assumption is similar to the
indivisibility assumption which is common in lower bounds on the I/O
complexity. 
My mimic the argument in~\cite{PaghS13} for triangle
enumeration, it can be proved that the enumeration requires $\BOM{{T}/({B M^{ k/2-1}})+{T^{2/k}}/{B}}$ I/Os. The claim follows by the fact that there cannot be more than $\BT{m^{k/2}}$ instances of a subgraph in the Alon class in a graph of $m$ edges~\cite{Alon81}. (For the sake of completeness we provide the proof in Appendix~\ref{app:lb}).
When $k=\BO{1}$,  the lower bound shows that our randomized algorithm is
optimal for any pattern graph, while the deterministic algorithm is optimal if
$s=k/2$ and $M=E^\epsilon$ for some constant $\epsilon>0$. Indeed, if the data
graph is a complete graph with $\sqrt{E}$ vertices, there exist $T=\BT{E^{k/2}}$
instances of any pattern graph with $k$ vertices.

\subsubsection{Work Complexity}
We analyze the work complexity when the pattern graph is in the Alon class and
$k=\BO{1}$. By using the ideas in~\cite[Theorem~6.2]{AfratiDU13}, the
enumeration (in internal memory) within each iteration
of the deterministic algorithm can be
performed in
$\BTO{M^{k/2-1}}$ work.
Then the total work of the deterministic algorithm is $\BTO{E^{k-s}/M^{k/2-s}}$.
As a consequence the expected work of the randomized algorithm becomes
$\BTO{E^{k/2}}$, which is just a polylog factor from the optimum since instances
in the Alon class (e.g., cliques) can appear $\BT{E^{k/2}}$ times in the worst
case.
To the best of our knowledge, the only algorithm for enumerating a generic
pattern graph which does not belong to the Alon class is a brute-force
approach. 
In this case,  the deterministic algorithm requires $\BTO{E^{k-s}}$ work since
Step~\ref{s4b} can be performed in $\BTO{M^{k-s-1}}$ work using the brute-force
approach; the expected work of the randomized algorithm then becomes
$\BTO{E^{k/2}M^{k/2-s}}$. In this case the work may become the main bottleneck
in a practical implementation.

\section{Conclusion}

The worst case complexities of our algorithms have an exponential dependency on the vertex
number $k$ of the pattern graph, and they are thus mainly of theoretical interest.
The lower bound shows that this is the best result  in the worst case  under
standard assumptions. However, some experiments~\cite{ParkSKP14} on related
MapReduce algorithms for triangle enumeration shows interesting performance and
seems to suggest that the analysis of our algorithms can be improved by
expressing the complexities as function of some properties of the input graph
(e.g., arboricity) or of the output. An output sensitive algorithm for triangle
enumeration has recently been proposed by Bj{\"o}rklund et
al.~\cite{BjorklundPWZ14} in the RAM model, however the problem remains open in
the external memory for the enumeration of an arbitrary subgraph as well as for
triangle enumeration.

\textbf{Acknowledgments.} The author would like to thank Rasmus Pagh and
Andrea Pietracaprina for useful discussions.

{\small
\bibliographystyle{splncs}
\bibliography{subgraphs}
}

\newpage
\section*{Appendix}

\subsection{Deterministic EM Algorithm when  $s=k/2$}\label{app:detalgspecialcase}
We now explain how  to extend the algorithm to the case $s=k/2$, that is
when all vertices in $V_H\setminus S$ are matched  with a vertex
in $S$. Note that in this case $k$ must be even. We recall that, with our
notation, $h_i$ is matched with $h_{k/2+i}$ under $S$ for each $1\leq i \leq
k/2$. Let $\Gamma_k$ denote the set of indexes of vertices in $V_H\setminus
h_{k/2}$  adjacent to $h_k$ (i.e., $i\in \Gamma_k$ if and only if $i<k/2$ and
$(h_i,h_{k})\in E_H$).

We observe that in the previous algorithm, the vertex set $V_{k}$ is empty in
Step~\ref{s4b} since $h_k$ is the probe vertex of $h_{k/2}$ and thus $V_k$
is not filled  in Step~\ref{s2}. If there are no
incomplete vertices in each chunk, then the previous algorithm can be fixed  by
filling $V_k$ in Step~\ref{s2} with vertices that are connected to a vertex in
$V_j$ for every $j\in \Gamma_k$. Indeed, these are the only possible values on
which $v_{k}$ can be mapped when all vertices $h_i$ with $1\leq i < k/2$ are
mapped onto vertices in $V_i$. This operation requires no I/Os since all
adjacency lists in each chunk are completely contained in internal memory, and
hence the upper bound in Theorem~\ref{th:det} still applies. Instead of proving
this claim, we propose a more general approach that holds even with incomplete
vertices.

Two major changes are required in the deterministic algorithm. The first one
allows to correctly enumerate all instances where at least one vertex in $\{h_i,
\forall i\in \Gamma_k\}$ is mapped onto a complete vertex. Then the second
change, which is more articulated, allows to enumerate all instances where each
vertex in $\{h_i, \forall i\in \Gamma_k\}$ is mapped onto an incomplete vertex. 

\emph{First change.} In Step~\ref{s2}, we add to $V_k$ all vertices $v$ which
are neighbors of complete vertices in $V_j$ for some $j\in \Gamma_k$.
Specifically, for each edge $(u,v)$ in $E_j$ with $j\in \Gamma_k$, $u\in
V_j$ and $u$ complete, vertex $v$ is added to $V_k$. As we will see in the main
proof, this
change allows to enumerate all instance where at least one vertex in $\{h_i,
\forall i\in \Gamma_k\}$ is mapped onto a complete vertex. For clearness,
consider the following example. Let $h_1$ be adjacent to $h_k$ and let be $v$ a
complete vertex in $V_1$. If $h_1$ is mapped onto $v$, the possible values onto
which $h_k$ can be mapped is given by the adjacency list of $v$ which is totally
in memory (thus, $V_k$ is set to these vertices). Then, for complete the
enumeration in the current round we have to insert into $E'$ each edge
connecting incomplete vertices in each $V_j$, for $j\in \Gamma_k$ with a vertex
in $V_k$ (this operation is performed by Step~\ref{s3} without further
modifications)

\emph{Second change.} 
We add a new operation before Step~\ref{s5}, but outside the iteration loop in
Step~\ref{s4}. This operation is performed only if there exists an incomplete
vertex in each chunk $C_{\ell_i}$ with $i\in \Gamma_k$ and let $v'_1, \ldots,
v'_{\Gamma_k}$ these vertices  (otherwise, there would no instances where each
vertex in $\{h_i, \forall i\in \Gamma_k\}$ is mapped onto an incomplete vertex
and this modification would be useless). The algorithm computes a set $V'$,
stored in external memory since it may exceed the internal memory size,
containing all vertices that are connected with all vertices $v'_1, \ldots
v'_{\Gamma_k}$ in $G$; this set can be computed by merging the adjacency lists
of $v'_1, \ldots v'_{\Gamma_k}$ and keeping only vertices that appear $\Gamma_k$
times. Then, using sorting, we compute a new edge list $\hat E$ containing all
edges with at least one extreme in $V'$. For each edge in $\hat E$, we call
\emph{linked} the vertex in $V'$. We denote with $\hat V$ the vertices in $\hat
E$ and require $\hat E$ to be stored as a collection of adjacency lists.
Subsequently, the algorithm enumerates all instances where vertex $h_{k/2}$ is
mapped onto a vertex in $\hat V$,  its probe edge onto an edge in $\hat E$,
$h_{k}$ onto the linked vertex of this edge (i.e., with a vertex in $V'$), and
$h_i$ on the incomplete vertex in $V_i$ for each  $i\in \Gamma_k$. The
enumeration is performed in iterations as in the previous algorithm, and in each
iteration the algorithm maps $h_{k/2}$ on a vertex $\hat V$ and its adjacency
list loaded in memory (if the adjacency list is too long we split it into
segments of size $M/(8k)$). (Clearly, as for every instance enumerated in the
current round, we also require that vertex $h_i$ is mapped onto a vertex in
$V_i$ and its probe edge onto $E_i$ for any $1\leq i \leq k/2-1$.) We observe
that it is not needed to load in memory the adjacency lists of incomplete
vertices in $\{V_i, \forall i\in \Gamma_k\}$ since each vertex in $V'$ is
connected with all of them by construction. The I/O complexity of the algorithm
is bounded by the following theorem. 
\begin{theorem}
The above algorithm correctly enumerates all instances of a
given pattern graph $H$ and its I/O complexity is 
$\BO{(8k)^{k-s-1}\frac{E^{k-s}}{B M^{k-s-1}} \log_M E}.$
\end{theorem}
\begin{proof}
Consider the first change. In Step~\ref{s3}, we load in memory each edge
between incomplete vertices in $\cup_{i\in \Gamma_k} V_i$ and vertices in
$V_k$. By mimic the proof of Theorem~\ref{th:det}, it can be shown that the
algorithm correctly enumerates all instances where vertex $h_i$ is mapped onto a
vertex in $V_i$, its probe edge onto $E_i$ for any $1\leq i \leq k/2-1$
\emph{and} at least one vertex in $\{h_i, \forall i\in \Gamma_k\}$ is mapped
onto a complete vertex. However, it may happen that instances where all 
vertices in $\{h_i, \forall i\in \Gamma_k\}$ are mapped onto incomplete vertices
are not enumerated since some edges could be missing in $E'$.

This is fixed by the second change. Indeed, the construction on $\hat E$
guarantees that for each $(u,v)\in \hat E$, where $v$ is marked as linked, the
vertex $v$ is connected to every incomplete vertex in $\{V_i, \forall i\in
\Gamma_k\}$. Therefore. as soon as $h_i$ is mapped on the incomplete vertex in
$V_i$, with $i\in \Gamma_k$, and $h_k$ is mapped onto $v$, we have that the
edge dependencies are correctly enumerated (even if edge information are not
currently available in internal memory).

Finally, we note that the first change does not increase the I/O complexity
and load in memory at most $\Gamma_k \cdot M/(2k)\leq M/4$ additional
edges/vertices. The second change requires $\BO{(E/B) \log_M E}$ I/Os per round
(i.e., sorting complexity) and load in memory at most $M/(8k)$ edges per
iteration. Since the space used by the first change can be deallocated before
the operations required by the second change start, the total amount of internal
memory never exceeds $M$
(recall as shown in the previous theorem, the base algorithm requires about
$M(k-1)/(2k)$ words of internal memory). The claimed I/O complexity easily
follows.
\end{proof}

We observe that the deterministic algorithm requires a MIS $S$ of the pattern graph (i.e., each vertex in $S$ is matched with a vertex
not in $S$) in order to correctly enumerate instances with incomplete
vertices. As an example consider the following case. Let the pattern graph be a
path of length 3, let $h_1$ be adjacent to vertices $h_{0}$ and $h_{2}$, and let
$S=\{h_0, h_2\}$ be a standard independent set of $H$ (note that it is not
matched). Suppose that there exists
an instance where vertices $v, v', v''$ are mapped onto $h_0, h_{1}, h_{2}$
respectively. If $v'$ is incomplete and edges $(v,v')$ and $(v',v'')$ are in
distinct chunks, then the two edges may not be at the same time in the internal
memory and then the instance cannot be emitted. This problem disappears if the
maximum degree of the input data graph is $\BO{M/k}$ since there are no
incomplete vertices and then all edges connected to a vertex are available
within a single chunk. In this case, it can be proved that the I/O complexity
reduces to $\BO{(8k)^{k-s'-1}E^{k-s'}/(BM^{k-s'-1})}$ I/Os, where $s'$ is the
size of a traditional independent set $S'$ of the pattern graph. This implies
that it is possible to go below the $\BO{E^{k/2}/(BM^{k/2-1})}$ bound if $s'\geq
k/2$, such as in stars, paths of odd length, or meshed with odd side. Clearly,
for these patterns the lower bound in Section~\ref{sec:ext} does not apply since
they are not in the Alon class.

\subsection{Lower bound on the I/O Complexity}\label{app:lb}
 The proof mimics the argument in~\cite{PaghS13} for triangle enumeration, but
exploits the fact that there cannot be more than $\BT{m^{k/2}}$ instances of a
subgraph in the Alon class in a graph of $m$ edges~\cite{Alon81}. 
The execution of an algorithm on a memory of size $M$ can be simulated, without
increasing the I/O complexity,  in a memory of size $2M$ so that the computation
proceeds in rounds. 
In each round (with the possible exception of the last round) there are
$\BT{M/B}$ I/Os, and memory blocks are read from (resp., written on) the
external memory only at the begin (resp., end) of a round. (We refer
to~\cite{PaghS13} for more details on the simulation.) 
By the aforementioned result on the Alon class, $\BT{M^{k/2}}$ instances can be
enumerated in a round since there are at most $2M$ edges in memory. Then, there
must be at least $\lfloor T/\BT{M^{k/2}}\rfloor$ rounds. Since each round needs
$\BT{M/B}$ I/Os, we get the first part of the claim.
The second term  follows since $\BOM{T^{2/k}}$ input edges must be read to
enumerate $T$ distinct instances.  
\qed

\subsection{Enumeration of Induced Subgraphs.}\label{app:induced}
The deterministic and randomized algorithms  can be easily adapted to enumerate
all \emph{induced} instances of a given subgraph. The I/O complexity of the
deterministic algorithm does increase asymptotically, while the I/O complexity
of the randomized algorithm shows only a small increase in the exponent of the
term $k^{\BO{k}}$.
It suffices to run the deterministic algorithm as the subgraph was a $k$-clique
$s=1$ and hence we can use the simple deterministic algorithm bounded in
Theorem~\ref{th:det}). In each
iteration, the algorithm contains all edges in $E$ between
any pair of vertices in  $\cup_{i=1}^k V_i$. Then, all instances of $H$ are
found, but only  induced instances are enumerated. This is possible since all
edges between vertices in the instance are available in memory. The  I/O
complexity of the algorithm  then  becomes $\BO{(8k)^{k-2} E^{ k-1}/(BM^{
k-2})}$. 
By using this algorithm for solving subproblems in the randomized algorithm, we
get an enumeration algorithm for induced subgraphs requiring  $\BO{(8k)^{4k}
E^{k/2}/(B M^{k/2-1})}$  I/Os, assuming that the maximum vertex degree is
$\sqrt{EM}$. The high probability result applies as well.
\end{document}

%% file: mis.tex
\begin{picture}(0,0)%
\includegraphics{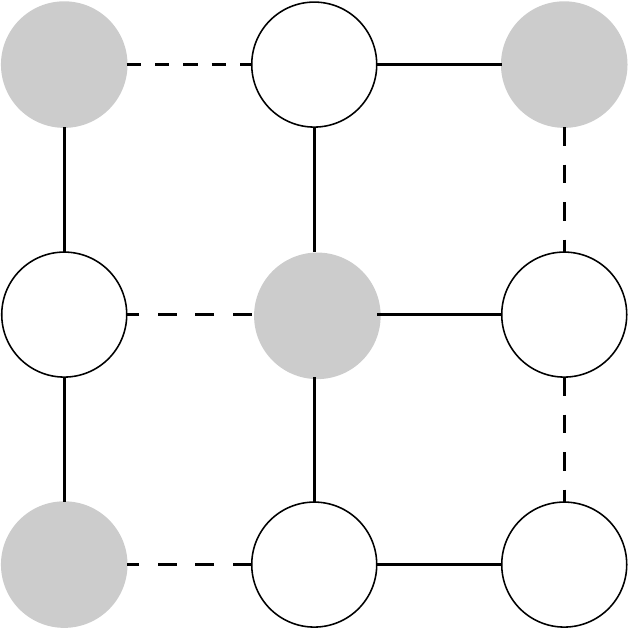}%
\end{picture}%
\setlength{\unitlength}{3947sp}%
\begingroup\makeatletter\ifx\SetFigFont\undefined%
\gdef\SetFigFont#1#2#3#4#5{%
  \reset@font\fontsize{#1}{#2pt}%
  \fontfamily{#3}\fontseries{#4}\fontshape{#5}%
  \selectfont}%
\fi\endgroup%
\begin{picture}(3016,3014)(2393,-6068)
\put(2701,-3436){\makebox(0,0)[b]{\smash{{\SetFigFont{12}{14.4}{\rmdefault}{\mddefault}{\updefault}{\color[rgb]{0,0,0}$h_6$}%
}}}}
\put(3901,-4636){\makebox(0,0)[b]{\smash{{\SetFigFont{12}{14.4}{\rmdefault}{\mddefault}{\updefault}{\color[rgb]{0,0,0}$h_7$}%
}}}}
\put(2701,-5836){\makebox(0,0)[b]{\smash{{\SetFigFont{12}{14.4}{\rmdefault}{\mddefault}{\updefault}{\color[rgb]{0,0,0}$h_9$}%
}}}}
\put(5101,-3436){\makebox(0,0)[b]{\smash{{\SetFigFont{12}{14.4}{\rmdefault}{\mddefault}{\updefault}{\color[rgb]{0,0,0}$h_8$}%
}}}}
\put(3901,-3436){\makebox(0,0)[b]{\smash{{\SetFigFont{12}{14.4}{\rmdefault}{\mddefault}{\updefault}{\color[rgb]{0,0,0}$h_1$}%
}}}}
\put(2701,-4636){\makebox(0,0)[b]{\smash{{\SetFigFont{12}{14.4}{\rmdefault}{\mddefault}{\updefault}{\color[rgb]{0,0,0}$h_2$}%
}}}}
\put(3901,-5836){\makebox(0,0)[b]{\smash{{\SetFigFont{12}{14.4}{\rmdefault}{\mddefault}{\updefault}{\color[rgb]{0,0,0}$h_4$}%
}}}}
\put(5101,-5836){\makebox(0,0)[b]{\smash{{\SetFigFont{12}{14.4}{\rmdefault}{\mddefault}{\updefault}{\color[rgb]{0,0,0}$h_5$}%
}}}}
\put(5101,-4636){\makebox(0,0)[b]{\smash{{\SetFigFont{12}{14.4}{\rmdefault}{\mddefault}{\updefault}{\color[rgb]{0,0,0}$h_3$}%
}}}}
\put(3301,-4486){\makebox(0,0)[b]{\smash{{\SetFigFont{8}{9.6}{\rmdefault}{\mddefault}{\updefault}{\color[rgb]{0,0,0}$P(2)=7$}%
}}}}
\put(3301,-3286){\makebox(0,0)[b]{\smash{{\SetFigFont{8}{9.6}{\rmdefault}{\mddefault}{\updefault}{\color[rgb]{0,0,0}$P(1)=6$}%
}}}}
\put(3301,-5686){\makebox(0,0)[b]{\smash{{\SetFigFont{8}{9.6}{\rmdefault}{\mddefault}{\updefault}{\color[rgb]{0,0,0}$P(4)=9$}%
}}}}
\put(5251,-3961){\rotatebox{270.0}{\makebox(0,0)[b]{\smash{{\SetFigFont{8}{9.6}{\rmdefault}{\mddefault}{\updefault}{\color[rgb]{0,0,0}$P(3)=8$}%
}}}}}
\put(5251,-5161){\rotatebox{270.0}{\makebox(0,0)[b]{\smash{{\SetFigFont{8}{9.6}{\rmdefault}{\mddefault}{\updefault}{\color[rgb]{0,0,0}$P(5)=3$}%
}}}}}
\end{picture}%

%% file: subgraphs.bbl
\begin{thebibliography}{10}

\bibitem{KairamWL12}
Kairam, S.R., Wang, D.J., Leskovec, J.:
\newblock The life and death of online groups: Predicting group growth and
  longevity.
\newblock In: Proc. 5th WSDM. (2012)  673--682

\bibitem{GregoriLM13}
Gregori, E., Lenzini, L., Mainardi, S.:
\newblock Parallel k-clique community detection on large-scale networks.
\newblock IEEE Trans. Paral. Dist. Systems \textbf{24}(8) (2013)  1651--1660

\bibitem{AfratiDU13}
Afrati, F.N., Fotakis, D., Ullman, J.D.:
\newblock Enumerating subgraph instances using {M}ap-{R}educe.
\newblock In: Proc. 29th ICDE. (2013)  62--73

\bibitem{PaghS13}
Pagh, R., Silvestri, F.:
\newblock The input/output complexity of triangle enumeration.
\newblock In: Proc. of 33rd PODS. (2014)  224--233

\bibitem{AfratiSSU13}
Afrati, F.N., Sarma, A.D., Salihoglu, S., Ullman, J.D.:
\newblock Upper and lower bounds on the cost of a {M}ap-{R}educe computation.
\newblock Proc. VLDB Endow. \textbf{6}(4) (2013)  277--288

\bibitem{GoodrichP11}
Goodrich, M., Pszona, P.:
\newblock External-memory network analysis algorithms for naturally sparse
  graphs.
\newblock In: Proc. 19th ESA. Volume 6942 of LNCS. (2011)  664--676

\bibitem{ChuC12}
Chu, S., Cheng, J.:
\newblock Triangle listing in massive networks.
\newblock ACM Trans. Knowl. Discov. Data \textbf{6}(4) (2012)  17:1--17:32

\bibitem{HuTC13}
Hu, X., Tao, Y., Chung, C.W.:
\newblock Massive graph triangulation.
\newblock In: Proc. ACM SIGMOD. (2013)  325--336

\bibitem{HuTQ15}
Hu, X., Qiao, M., Tao, Y.:
\newblock Join dependency testing, loomis-whitney join, and triangle
  enumeration.
\newblock In: Proc. 34th ACM PODS. (2015)  291--301

\bibitem{ChibaN85}
Chiba, N., Nishizeki, T.:
\newblock Arboricity and subgraph listing algorithms.
\newblock SIAM J. Comput. \textbf{14}(1) (1985)  210--223

\bibitem{NgoPRR12}
Ngo, H.Q., Porat, E., R{\'e}, C., Rudra, A.:
\newblock Worst-case optimal join algorithms.
\newblock In: Proc. 31st PODS. (2012)  37--48

\bibitem{KolountzakisMPT12}
Kolountzakis, M.N., Miller, G.L., Peng, R., Tsourakakis, C.E.:
\newblock Efficient triangle counting in large graphs via degree-based vertex
  partitioning.
\newblock Internet Mathematics \textbf{8}(1-2) (2012)  161--185

\bibitem{KaneMSS12}
Kane, D.M., Mehlhorn, K., Sauerwald, T., Sun, H.:
\newblock Counting arbitrary subgraphs in data streams.
\newblock In: Proc. 39th ICALP. (2012)  598--609

\bibitem{WilliamsW13}
Vassilevska~Williams, V., Williams, R.:
\newblock Finding, minimizing, and counting weighted subgraphs.
\newblock SIAM J. Comput. \textbf{42}(3) (2013)  831--854

\bibitem{FinocchiFF14}
Finocchi, I., Finocchi, M., Fusco, E.G.:
\newblock Counting small cliques in mapreduce.
\newblock arXiv:1403.0734 (2014)

\bibitem{SuriV11}
Suri, S., Vassilvitskii, S.:
\newblock Counting triangles and the curse of the last reducer.
\newblock In: Proc. 20th ACM International Conference on World Wide Web. (2011)
   607--614

\bibitem{ParkSKP14}
Park, H.M., Silvestri, F., Kang, U., Pagh, R.:
\newblock Mapreduce triangle enumeration with guarantees.
\newblock In: Proc. 23rd CIKM. (2014)

\bibitem{Vitter08}
Vitter, J.S.:
\newblock Algorithms and Data Structures for External Memory.
\newblock Now Publishers Inc. (2008)

\bibitem{Janson04}
Janson, S.:
\newblock Large deviations for sums of partly dependent random variables.
\newblock Random Structures \& Algorithms \textbf{24}(3) (2004)  234--248

\bibitem{GradwohlY08}
Gradwohl, R., Yehudayoff, A.:
\newblock t-wise independence with local dependencies.
\newblock Information Processing Letters \textbf{106}(5) (2008)  208 -- 212

\bibitem{Alon81}
Alon, N.:
\newblock On the number of subgraphs of prescribed type of graphs with a given
  number of edges.
\newblock Israel Journal of Mathematics \textbf{38}(1-2) (1981)  116--130

\bibitem{BjorklundPWZ14}
Bj{\"o}rklund, A., Pagh, R., Williams, V.V., Zwick, U.:
\newblock Listing triangles.
\newblock In: Proc. of 41st ICALP. Volume 8572 of LNCS. (2014)  223--234

\end{thebibliography}
